\documentclass[11pt,a4paper,twoside]{article}
\usepackage{a4wide, amssymb, amsmath, amsthm, graphics, comment, xspace, enumerate}
\usepackage{graphicx,multirow}
\usepackage[noend]{algpseudocode}
\usepackage{caption}
\usepackage{tabularx}
\usepackage{enumitem, multirow}
\usepackage[usenames,dvipsnames]{color}
\usepackage{alltt}
\usepackage{color}
\usepackage{listings}
\usepackage{cite}

\usepackage{float}
\usepackage{xspace}
\usepackage[section]{algorithm}
\usepackage[noend]{algpseudocode}

\definecolor{string}{rgb}{0.7,0.0,0.0}
\definecolor{comment}{rgb}{0.13,0.54,0.13}
\definecolor{keyword}{rgb}{0.0,0.0,1.0}
\usepackage{amsmath,amsfonts,amssymb,subfigure}
\usepackage{tikz}
\usetikzlibrary{arrows}
\usetikzlibrary{decorations}
\usetikzlibrary{patterns}

\usepackage{booktabs} 

\subfiguretopcaptrue
\tikzstyle{vtx}=[circle, inner sep= 0pt, minimum size= 1.2mm, fill]

\newtheorem{te}{Theorem}[section]
\newtheorem{pro}[te]{Proposition}
\newtheorem{de}{Definition}[section]

\newtheorem{lemma}[te]{Lemma}

\newtheorem{conjecture}{Conjecture}[section]

\newtheorem{re}{Remark}[section]

\newcommand{\beq}{\begin{eqnarray}}
\newcommand{\eeq}{\end{eqnarray}}

\newcommand{\beqs}{\begin{eqnarray*}}
\newcommand{\eeqs}{\end{eqnarray*}}

\newcommand{\ABC}{{\rm ABC}}

\newcommand{\ds}{\displaystyle}
\allowdisplaybreaks

\begin{document}
\title{ On structural properties of trees \\with minimal atom-bond connectivity index  IV: \\
Solving a conjecture about the pendent paths of length three}
\maketitle
\begin{center}
{\large \bf  Darko Dimitrov}
\end{center}
\baselineskip=0.20in
\begin{center}
{\it Hochschule f\"ur Technik und Wirtschaft Berlin, Germany \& \\
     Faculty of Information Studies, Novo Mesto, Slovenia} 
\\E-mail: {\tt darko.dimitrov11@gmail.com}
\end{center}

\baselineskip=0.20in
\vspace{6mm}
\begin{abstract}
The atom-bond connectivity (ABC) index is one of the 
most investigated degree-based molecular structure descriptors
with a variety of chemical applications.
It is known that among all connected graphs, the trees minimize the ABC index.
However, a full characterization of trees with a minimal ABC index is still an open problem.
By now, one of the proved properties is that a tree with a minimal ABC index may have, 
at most, one pendent path of length $3$, with the conjecture  that it cannot be a case  if the order of a tree is
larger than $1178$.
Here, we provide an affirmative answer of a strengthened version of that conjecture, 
showing that a tree with minimal ABC index cannot contain
a pendent path of length $3$ if its order is
larger than $415$.

\end{abstract}
\medskip
%
%
%
\section[Preliminaries and related results]{Preliminaries and related results}

Let $G=(V, E)$ be a simple undirected graph of order $n=|V|$ and size $m=|E|$.
For $v \in V(G)$, the degree of $v$, denoted by $d_G(v)$, is the number of edges incident
to $v$. 
When it is clear from the context we will write $d(v)$, which will always assume $d_G(v)$.
The atom-bond connectivity index of $G$ is defined as
\beq \label{eqn:001}
\ABC(G)=\sum_{uv\in E(G)}\sqrt{\frac{(d(u) +d(v)-2)}{d(u)d(v)}}= \sum_{uv\in E(G)}f(d(u),d(v)).
\eeq
The ABC index was introduced in 1998 by Estrada, Torres, Rodr{\' i}guez and Gutman~\cite{etrg-abc-98} and
is one of the most investigated  degree-based molecular structure descriptors.
More about the (degree-based) molecular structure descriptors  can be found in~\cite{fgd-ssdbti-13, g-dbti-13, tc-ndc-09}  and in the references cited therein.
In ~\cite{etrg-abc-98} it was shown that the ABC index can be a valuable predictive tool in the study of the heat of formation in alkanes. 
Additional physico-chemical applicabilities of the ABC index were presented in few other works including~\cite{e-abceba-08, dt-cbfgaiabci-10,  gtrm-abcica-12}.
These triggered a series of works that considered both mathematical and computational aspect of the ABC index~\cite{adgh-dctmabci-14, ahs-tmabci-13, ahz-ltmabci-13, cg-eabcig-11, clg-subabcig-12, dmga-cbabcig-16, d-abcig-10, d-ectmabci-2013, d-sptmabci-2014, d-sptmabci-2-2015, ddf-sptmabci-3-2016, dgf-abci-11, fgiv-cstmabci-12, fgv-abcit-09, ftvzag-siabcigo-2011, ghl-srabcig-11, gly-abctgds-12, gs-sabcitnpv-16, gfahsz-abcic-2013, gmng-abcitfnl-15, lccglc-fcstmaibtds-14, lcmzczj-twmabciatgnl-16, lmccgc-pstmabci-15, p-rubabci-14, vh-mabcict-2012, xzd-frabcit-2010}.

It is known that among all connected  graphs with $n$ vertices, the graph with minimal ABC index is a tree~\cite{dgf-abci-11,cg-eabcig-11}.
Although there is some significant progress in characterizing the trees with minimal ABC index (also refereed as  minimal-ABC trees), the complete
characterization is still open.

Before we state the main contribution of this paper,
we first present some additional notation as used in the rest of the paper.
A vertex of degree one is a {\it pendent vertex}.
A vertex is {\it big}, if its degree is at least $3$ and it is not adjacent to a vertex of degree $2$.
As in \cite{gfi-ntmabci-12}, a sequence of vertices of a graph $G$, $S_k=v_0 \, v_1 \dots v_k$,  will be  called a {\it pendent path} if 
each two consecutive vertices in $S_k$ are adjacent in $G$, $d(v_0)>2$, $d(v_i)=2$, for $i=1,\dots k-1$, and $d(v_k)=1$.
The length of the pendent path $S_k$ is $k$.
If $d(v_k) > 2$, then $S_k$ is an {\it internal path} of length $k$.
A {\it proper Kragujevac tree} \cite{hag-ktmabci-14} is a tree possessing a central vertex of degree
at least $3$, to which branches $B_k$-branches, $k \geq 1$  are
attached (see Figure for an illustration of $B_k$- and $B^*_k$-branches). 
By inserting a new vertex (of degree $2$) on a pendent edge in a $B_k$-branch, we obtain a $B^*_k$-branch.
An {\it improper Kragujevac tree} is a tree obtained from  a Kragujevac tree by replacing one $B_k$-branch with 
a $B^*_k$-branch.
\begin{figure}[!h]
\begin{center}
\includegraphics[scale=0.9]{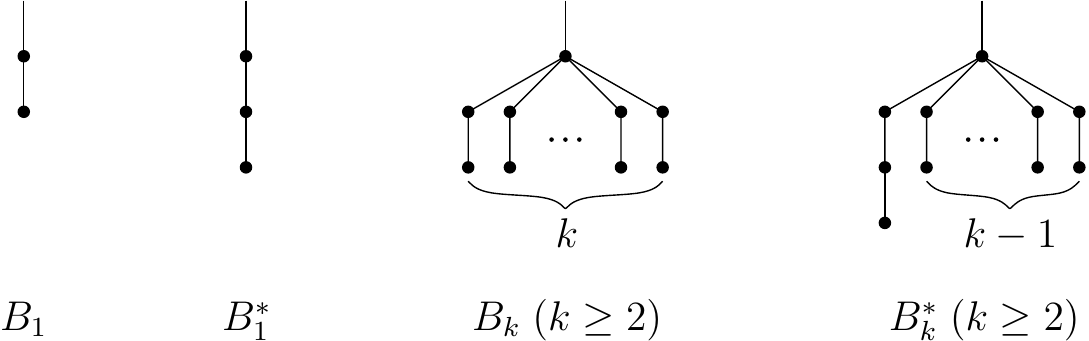}
\caption{An illustration of $B_k$- and $B^*_k$-branches.}
\label{B_k-branches}
\end{center}
\end{figure}
In~\cite{w-etwgdsri-2008} Wang defined a {\em greedy tree} as follows.

\begin{de}\label{def-GT}
Suppose the degrees of the non-leaf vertices are given, the greedy tree is achieved by the following `greedy algorithm':
\begin{enumerate}
\item Label the vertex with the largest degree as $v$ (the root).
\item Label the neighbors of $v$ as $v_1, v_2,\dots,$ assign the largest degree available to them such that $d(v_1) \geq d(v_2) \geq \dots$
\item Label the neighbors of $v_1$ (except $v$) as $v_{11}, v_{12}, \dots$ such that they take all the largest
degrees available and that $d(v_{11}) \geq d(v_{12}) \geq . . .$ then do the same for $v_2, v_3,\dots$
\item Repeat 3. for all newly labeled vertices, always starting with the neighbors of the labeled vertex with largest whose neighbors are not labeled yet.
\end{enumerate}
\end{de}
\noindent

\noindent
In the sequel few structural properties of the minimal-ABC trees relevant to the result of this work are presented. For all other known properties of the  minimal-ABC trees, we referee to \cite{adgh-dctmabci-14, d-ectmabci-2013, d-sptmabci-2014, d-sptmabci-2-2015, ddf-sptmabci-3-2016, df-ftmabc-2015, gfahsz-abcic-2013}.

\subsection[Related results]{Related results}

\noindent
The following result  characterizes the trees with minimal ABC index with prescribed degree sequences. 

\begin{te}[\cite{gfi-ntmabci-12, xz-etfdsabci-2012}]\label{thm-DS}
Given the degree sequence, the greedy tree minimizes the ABC index.
\end{te}

\smallskip
\noindent
The following three results reveal some properties of the paths in the minimal-ABC trees.

\begin{te}[\cite{gfi-ntmabci-12}]\label{thm-GFI-10}
The $n$-vertex tree with minimal ABC index does not
contain internal paths of any length $k \geq 1$.
\end{te}

\begin{te}[\cite{llgw-pcgctmabci-13}] \label{thm-LG-10}
Each pendent vertex of an $n$-vertex tree with minimal
ABC index belongs to a pendent path of length $k $, $2 \leq k \leq 3$.
\end{te}

\begin{te}[\cite{gfi-ntmabci-12}] \label{thm-GFI-30}
The $n$-vertex tree with minimal ABC index contains at
most one pendent path of length $3$.
\end{te}

\noindent
In the context of pendent paths, we assume the following level representation of a greedy tree:
The root of a greedy tree  belongs to the highest level $i$ of the tree.
If so, the leaves of a greedy tree can belong to level $1$ or $2$, except 
the pendant vertex of a pendent path of length $3$ that may belong to levels $0$ or $1$ 
(see the graph $G$ in Figure~\ref{fig-P3-20} for an illustration).

\smallskip
\noindent
The next three results present some conditions on the occurrence of a pendent path of length $3$  in the minimal-ABC trees.
Recall that a $B_k^*$-branch is obtained from a $B_k$-branch by replacing one pendent path of length $2$ with a pendent path of length $3$.

\begin{lemma}[\cite{df-ftmabc-2015}]  \label{le-noBk-star}
A minimal-ABC tree does not contain a $B_k^*$-branch, $k \geq 4$.
\end{lemma}

\begin{te}[\cite{ddf-sptmabci-3-2016}]  \label{no-B1star}
A minimal-ABC tree of order $n >18$ with a pendent path of length $3$ does not contain $B_1$-branch ($B_1^*$-branch).
\end{te}

\begin{te} [\cite{ddf-sptmabci-3-2016}]  \label{no-B2star}
A minimal-ABC tree  of order $n > 18$ with a pendent path of length $3$ may contain a $B_2$-branch if and only if it is of order $161$ or $168$.
Moreover, in this case a minimal-ABC tree is comprised of single central vertex, $B_3$-branches and one $B_2$, including a
pendent path of length $3$ that may belong to a $B_3^*$-branch or $B_2^*$-branch.
\end{te}

\noindent
In \cite{d-ectmabci-2013} the following important conjecture was raised.

\begin{conjecture} \label{conj-P3}
A minimal ABC tree of order $n >1178$ does not contain a pendent path of length three.
\end{conjecture}

\noindent
In the next section, in Theorem~\ref{te-P3-20}, we prove a stronger version of the above conjecture  by showing that 
a minimal-ABC tree of order $n >415$ does not contain a pendent path of length three.
In the appendix we present some auxiliary results that will be used to prove the conjecture.

\section[The proof of Conjecture~\ref{conj-P3}]{The proof of Conjecture~\ref{conj-P3}}\label{sec:Results}

\noindent
The following proposition will be used in the proof of the main result of this paper, Theorem~\ref{te-P3-20}.

\begin{pro} \label{pro-P3-10}
Let $x$ und $y$ be vertices of a minimal-ABC tree $G$ that have a common parent vertex $z$,  such that $d(x) \geq d(y) \geq 5$.
If $x$ and $y$ have only $B_3$-branches as children, then either $d(y) = d(x)$ or  $d(y) = d(x)-1$.
\end{pro}

\begin{proof}

Assume that the proposition is false and that $d(x)+\Delta = d(y)$, where $\Delta\geq 2$.
Apply the transformation $\mathcal{T}$ illustrated in Figure~\ref{fig-P3-10}.
\begin{figure}[!ht]
\begin{center}
\includegraphics[scale=0.75]{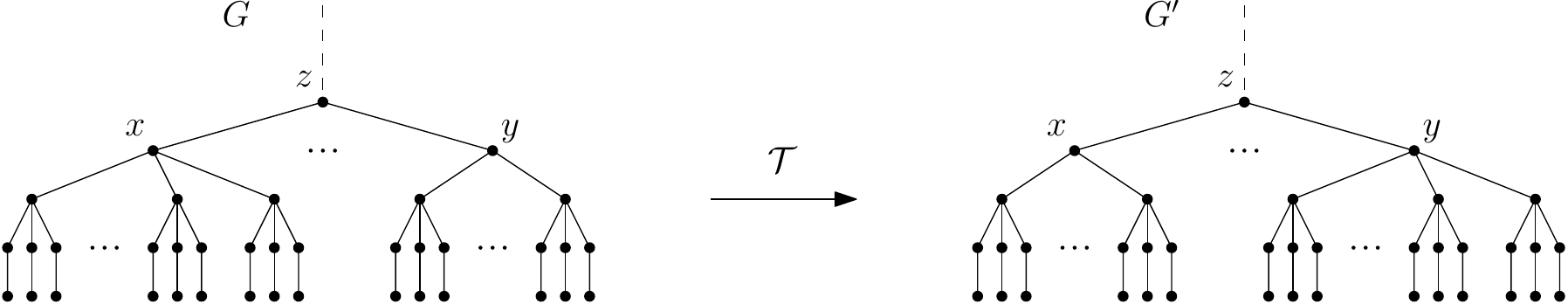}
\caption{The transformation $\mathcal{T}$ from the proof of  Proposition~\ref{pro-P3-10}.}
\label{fig-P3-10}
\end{center}
\end{figure}
After applying $\mathcal{T}$ the degree of the vertex $x$ decreases by $1$,
while the degree of $y$ increases by $1$.
The rest of the vertices do not change their degrees.
The change of the ABC index is then
\begin{eqnarray*} \label{eq-pro-P3-10-00}
&&-f(d(z),d(x))+f(d(z),d(x)-1)   \nonumber \\
&& -f(d(z),d(y))+f(d(z),d(y)+1) \nonumber \\
&&+(d(x)-2)(-f(d(x),4)+f(d(x))-1,4))  \nonumber \\
&&  -f(d(x),4)+f(d(y)+1,4)  \nonumber  \\
&& +(d(y)-1)(-f(d(y),4)+f(d(y)+1,4)), \nonumber 
\end{eqnarray*}
or, by setting $d(x)=d(y)+\Delta$, we obtain that the change of the ABC index is
\begin{eqnarray*} \label{eq-pro-P3-10}
g_1(d(z), \Delta, d(y)) &=&-f(d(z),d(y)+\Delta)+f(d(z),d(y)+\Delta-1)  \nonumber \\
 && -f(d(z),d(y))+f(d(z),d(y)+1) \nonumber \\
&&+(d(y)+\Delta-2)(-f(d(y)+\Delta,4)+f(d(y)+\Delta-1,4))  \nonumber \\
&&  -f(d(y)+\Delta,4)+f(d(y)+1,4) \nonumber  \\
&& +(d(y)-1)(-f(d(y),4)+f(d(y)+1,4)). \nonumber 
\end{eqnarray*}

\noindent
By Proposition~\ref{appendix-pro-030} (see the appendix), 
the expression $-f(d(z),d(y)+\Delta)+f(d(z),d(y)+\Delta-1)$ increases in $d(z)$, 
while  by Proposition \ref{appendix-pro-030-2}, $-f(d(z),d(y)+1))+f(d(z),d(y))$
decreases in $d(z)$. It follows that  for fixed $\Delta \geq 2$, $-f(d(z),d(y)+1))+f(d(z),d(y))$
increases faster in $d(z)$ than $-f(d(z),d(y)+\Delta)+f(d(z),d(y)+\Delta-1)$, 
or  expressed alternatively, $-(-f(d(z),d(y)+1))+f(d(z),d(y)))=-f(d(z),d(y))+f(d(z),d(y)+1))$
decreases faster in $d(z)$ than $-f(d(z),d(y)+\Delta)+f(d(z),d(y)+\Delta-1))$ increases in $d(z)$.
Consequently, $g_1(d(z), \Delta, d(y))$ decreases in $d(z)$, and it is maximal when $d(z)$ is minimal.
By Theorem~\ref{thm-DS}, it follows that the minimal possible value of $d(z)$ is  $d(x)=d(y)+\Delta$.

Now, we show that $g_1(d(z), \Delta, d(y))$ decreases in $\Delta$.
By Propositions~\ref{appendix-pro-030} and \ref{appendix-pro-030-2}, 
expressions $-f(d(z),d(y)+\Delta)+f(d(z),d(y)+\Delta-1)$ and   $-f(d(y)+\Delta,4)+f(d(y)+1,4)$  decrease in $\Delta$.
By Proposition~\ref{appendix-pro-040} (see the appendix) it follows that $(d(y)+\Delta-2)(-f(d(y)+\Delta,4)+f(d(y)+\Delta-1,4))$ also decreases in $\Delta$.
Therefore, $g_1(d(z), \Delta, d(y))$ too decreases in  $\Delta$, and the upper bound on $g_1(d(z), \Delta, d(y))$ is obtained for
$\Delta=2$ (which is the minimal value of $\Delta$ due to the assumptions of the proposition).

Bearing these in mind, we obtain that $g_1(d(z), \Delta, d(y))$ has the following upper bound
\begin{eqnarray*} \label{eq-pro-P3-30}
g_1(d(y)+2, 2, d(y)) = g_2(d(y))&=&-f(d(y)+2,d(y)+2)+f(d(y)+2,d(y)+1)  \nonumber \\
 && -f(d(y)+2,d(y))+f(d(y)+2,d(y)+1) \nonumber \\
&&+d(y)(-f(d(y)+2,4)+f(d(y)+1,4))  \nonumber \\
&&  -f(d(y)+2,4)+f(d(y)+1,4) \nonumber  \\
&& +(d(y)-1)(-f(d(y),4)+f(d(y)+1,4)). \nonumber 
\end{eqnarray*}
The function $g_2(d(y))$, $d(y) >4$, has only one zero at $d(y)=4.04954$, and for $d(y) > 4.04954$ is negative.
Thus, for $d(y) \geq 5$, $\Delta \geq 2$,  
the change of the ABC index after 
applying the transformation $\mathcal{T}$ is negative, which is a contradiction to the initial assumption that $G$ is a minimal-ABC tree.

\end{proof}

\noindent
Modifying the above proof by considering that $\Delta  \geq 3$, the following result can be obtained.
\begin{re}
The function $g_1(d(z), \Delta, d(y))$ from Proposition~\ref{pro-P3-10}  is negative for $\Delta  \geq 3$ and $d(y) \geq 4$.
\end{re}

\noindent
And, we can obtain a modified version of Proposition~\ref{pro-P3-10}.

\begin{pro} \label{pro-P3-20}
Let $x$ und $y$ be vertices of a minimal-ABC tree $G$ that have a common parent vertex $z$,  such that $d(x) \geq d(y) \geq 4$.
If $x$ and $y$ have only $B_3$-branches as children, then either $d(y)=d(x)$, $d(y)=d(x)-1$ or  $d(y)=d(x)-2$.
\end{pro}

\noindent
Next, we present the main result of this paper, that gives an affirmative answer to 
Conjecture~\ref{conj-P3}.

\begin{te} \label{te-P3-20}
A minimal-ABC tree of order $n > 415$ does not contain a pendent path of length three.
\end{te}
\begin{proof}
Denote a minimal-ABC tree of order $n \geq 415$ by $G$ and its root vertex by $z$.
Assume that the claim of the theorem
is false and there is one $P_3$ path in $G$ (i.e., one $B^*_3$-branch). 
Assume also that the leaf of the only path of length $3$ belongs to level $0$ of $G$.
By Lemma~\ref{le-noBk-star} and Theorems~\ref{thm-GFI-30},~\ref{no-B1star}~and~\ref{no-B2star}, it follows that the only type of branches
that $G$ can have are
$B_3$-branches and maybe one $B^*_3$-branch. 
We distinguish three main cases with respect to the maximal number of levels of $G$, or
with other words, with respect to that to which level of $G$ the root vertex $z$ belongs.

\bigskip

\noindent
{\bf Case $1.$}   {\it The root vertex $z$ of $G$ is at level $\geq 6$.}

\noindent
As consequence of Theorem~\ref{thm-DS}, all vertices at level $5$ are big, and at least $d(z)-1$ of them
do not have  any $B_3$-branch as child.
Denote this set with at least $d(z)-1$ vertices by $\mathcal{L}_5$.
If $d(z)=4$ then, by the same theorem, all vertices at level $5$  have degree $4$. 
If $d(z) \geq 5$, then consider one vertex $y$ from $\mathcal{L}_5$.
Let $x$ be child of $y$ with the smallest degree. 
We denote the other children of $y$ by  $x_i, i=1, \dots, d(y)-2$.
Without loss of generality, we may assume that $x$ is a parent of a $B_3^*$-branch.

\noindent
Here, we first apply the transformation $\mathcal{T}_1$ 
depicted  in 
Figure~\ref{fig-P3-20}.

\begin{figure}[!h]
\begin{center}
\includegraphics[scale=0.75]{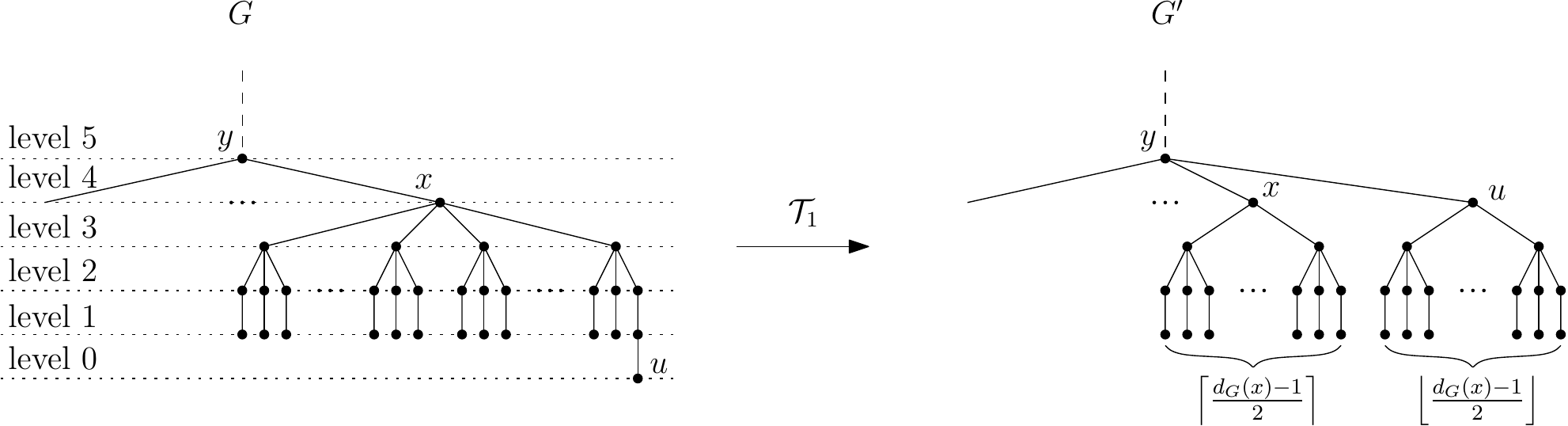}
\caption{The transformation $\mathcal{T}_1$ from the proof of  Theorem~\ref{te-P3-20}, Case $1.$}
\label{fig-P3-20}
\end{center}
\end{figure}
After applying $\mathcal{T}_1$ the degree of the vertex $y$ increases by $1$,
the degree of $x$ decreases to $\lceil ( d(x)-1)/2) \rceil  +1$
while the degree of $u$ increases form $1$ to $\lfloor ( d(x)-1)/2) \rfloor  +1$.
The rest of the vertices do not change their degrees.
The change of the ABC index is smaller than
\begin{eqnarray} \label{eq-pro-P3-110}
g_1(d(x), d(y), d(x_i)) =
&&-f(d(y),d(x))+f \left(d(y)+1,\left\lceil \frac{d(x)-1}{2} \right \rceil  +1\right)   \nonumber \\
&&-f(2,1)+f\left(d(y)+1,\left\lfloor \frac{d(x)-1}{2} \right \rfloor  +1\right)  \nonumber \\
&& +\sum_{i=1}^{d(y)-2}(-f(d(y),d(x_i))+f(d(y)+1, d(x_i))  \nonumber \\
&& + \left( \left\lceil \frac{d(x)-1}{2} \right \rceil \right) \left(-f\left(d(x),4\right) + f\left(\left\lceil \frac{d(x)-1}{2} \right \rceil  +1,4\right) \right) \nonumber \\
&& + \left( \left\lfloor \frac{d(x)-1}{2} \right \rfloor \right)  \left(-f\left(d(x),4\right) + f\left(\left\lfloor \frac{d(x)-1}{2} \right \rfloor  +1,4\right) \right).  \nonumber 
\end{eqnarray}
We would like to note that we slightly abuse the notation of a multivariable function in the case of the function $g_1$ (as well as in some later examples).
Namely, instead of $g_1(d(x), d(y), d(x_1),$ $ \dots, d(x_{d(y)-2}))$, we write $g_1(d(x), d(y), d(x_i))$.

By Proposition~\ref{appendix-pro-030} $-f(d(y),d(x_i))+f(d(y)+1, d(x_i)$ decreases in $d(x_i)$.
$d(x_i)$ can be not smaller than $d(x)$, since $x$ is the child of $y$ with the smallest degree.
So, by setting $d(x_i)=d(x)$, we obtain the upper bound of $-f(d(y),d(x_i))+f(d(y)+1, d(x_i)$.
Together with the fact that  $f(x,y)$ is a monotonically decreasing function in $x$ and $y$,  and by Proposition~\ref{appendix-pro-030-2},
 we have that for odd $d(x)$, $g_1(d(x), d(y), d(x_i))$
is bounded from above by
\begin{eqnarray} \label{eq-pro-P3-110-40}
g_{1o}(d(x), d(y)) =
&&-f(d(y),d(x))+f \left(d(y)+1,\frac{d(x)}{2} +0.5\right)  \nonumber \\
&& -f(2,1)+f\left(d(y)+1,\frac{d(x)}{2} +0.5 \right)   \nonumber \\
&& +(d(y)-2)(-f(d(y),d(x))+f(d(y)+1, d(x))  \nonumber \\
&& + \left( d(x) -1\right) \left(-f\left(d(x),4\right) + f\left( \frac{d(x)}{2}+0.5,4\right) \right) \nonumber 
\end{eqnarray}
Consider the
function ${\hat g}_{1o}(d(x), d(y))$ comprised of two expressions of $g_{1o}(d(x), d(y))$:
$$
{\hat g}_{1o}(d(x), d(y))= f\left(d(y)+1,\frac{d(x)}{2} +0.5 \right) +  \left( d(x) -1\right) \left(-f\left(d(x),4\right) + f\left( \frac{d(x)}{2}+0.5,4\right) \right). 
$$
The first derivative of ${\hat g}_{1o}(d(x), d(y))$ with respect of $d(x)$ is

\begin{eqnarray} \label{eq-pro-P3-110-50}
 \frac{ \partial {\hat g}_{1o}(d(x), d(y))}{ \partial d(x)} &=& \frac{1}{2} \left(\sqrt{\frac{2.5 +0.5 d(x)}{0.5 +0.5 d(x)}}-\sqrt{\frac{2+d(x)}{d(x)}}  \right.  \nonumber \\
 && \left. +\frac{1}{2} (d(x)-1) \left(-\frac{1.}{(0.5 +0.5 d(x))^2  \ds \sqrt{\frac{2.5 +0.5 d(x)}{0.5 +0.5 d(x)}}}+ \frac{2}{d(x)^2 \ds \sqrt{\frac{2+d(x)}{d(x)}}}  \right)  \right.  \nonumber \\
 && \left. + \ds \frac{0.5 -0.5 d(y)}{(0.5 +0.5 d(x))^2 (1 +d(y))  \ds \sqrt{\frac{-0.5+0.5 d(x)+d(y)}{(0.5 +0.5 d(x)) (1 +d(y))}}}   \right)
\end{eqnarray}
The last expression  in (\ref{eq-pro-P3-110-50})  decreases in $d(y)$,
so  $\partial {\hat g}_{1o}(d(x), d(y)) / \partial d(x)$ is maximal for $d(y)=d(x)$ (the minimal value of $d(y)$).
$\partial {\hat g}_{1o}(d(x), d(x)) / \partial d(x)$ for $d(x) \geq 5$ is a negative function, from which follows that
${\hat g}_{1o}(d(x), d(y))$ for $d(x) \geq 5$ decreases in $d(x)$.
By Proposition~\ref{appendix-pro-030-2}, the expressions $-f(d(y),d(x))+f \left(d(y)+1,\frac{d(x)}{2} +0.5\right) $ and $-f(d(y),d(x))+f(d(y)+1, d(x)-2$
also decrease in $d(x)$. Thus, it follows that  $g_{1o}(d(x), d(y))$ decreases in $d(x)$ for $d(x) \geq 5$
and therefore $g_{1o}(d(x), d(y))$ is bounded from above by $g_{1o}(5, d(y))= \bar g_{1o}(d(y))$.
The function $\bar g_{1o}((d(y))$ does not have real roots, it is positiv and has a horizontal asymptote at $0.215937$.
Thus, we can conclude that $g_{1o}(d(x), d(y))$ increases in $d(y)$ and $\lim_{d(y) \to \infty} g_{1o}(d(x), d(y))$ is an upper bound 
of $g_{1o}(d(x), d(y))$. The first integer value of $d(x)$ for which $\lim_{d(y) \to \infty} g_{1o}(d(x), d(y))$ is $54$. Since $g_{1o}(d(x), d(y))$
decreases in $d(x)$, it follows that the change of the ABC index after apply $\mathcal{T}_1$, is negative for $d(x) \geq 54$, when $d(x)$ is odd.

For even $d(x)$, $g_1(d(x), d(y), d(x_i))$
is bounded from above by
\begin{eqnarray} \label{eq-pro-P3-110-20}
g_{1e}(d(x), d(y)) =
&&-f(d(y),d(x))+f \left(d(y)+1,\frac{d(x)}{2} +1\right)  \nonumber \\
&&  -f(2,1)+f\left(d(y)+1,\frac{d(x)}{2} \right)   \nonumber \\
&& +(d(y)-2)(-f(d(y),d(x))+f(d(y)+1, d(x))  \nonumber \\
&& + \left( \frac{d(x)}{2} \right) \left(-f\left(d(x),4\right) + f\left( \frac{d(x)}{2}+1,4\right) \right) \nonumber \\
&& + \left( \frac{d(x)}{2} -1\right) \left(-f\left(d(x),4\right) + f\left( \frac{d(x)}{2},4\right) \right). \nonumber 
\end{eqnarray}
Almost an identical derivation, as in the case when $d(x)$ is odd, also here leads to the fact that  after applying the transformation $\mathcal{T}_1$
the change of the ABC index is negative for $d(x) \geq 54$. Thus, the derivation for the case when $d(x)$ is even we will be omitted.

\smallskip
\noindent
For  $d(x) \leq 53$ we apply the transformation $\mathcal{T}_2$ illustrated in Figure~\ref{fig-P3-30}. 
Let $w$ be a vertex that has the same parent as $x$. Recall that $x$ is vertex with smallest degree
among all vertices that have same parent as $x$.
If $d(x)=4$, then by Proposition~\ref{pro-P3-20} either $d(w)=d(x), d(w)=d(x)+1$ or $d(x)+2$.
Otherwise, by Proposition~\ref{pro-P3-10} either $d(w)=d(x)$ or $d(w)=d(x)+1$.
Let $y_p$ be the parent of $y$ and $x_i, i =1, \dots, d(y)-3$, the children vertices of $y$ different than $x$ and $w$.

\begin{figure}[!h]
\begin{center}
\includegraphics[scale=0.75]{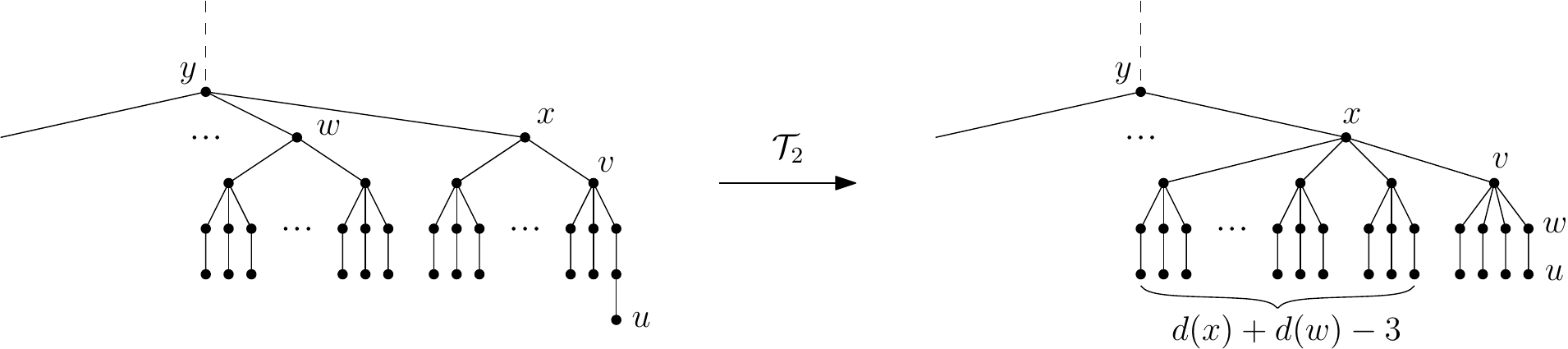}
\caption{The transformation $\mathcal{T}_2$ from the proof of  Theorem~\ref{te-P3-20}, Case $1.$}
\label{fig-P3-30}
\end{center}
\end{figure}

After applying $\mathcal{T}_2$ the degree of the vertex $x$ increases by $d(w)-1$,
the degree of $v$  increases by $1$, 
the degree of $y$ decreases by $1$,
while the degree of $w$ decreases to $2$.
The rest of the vertices do not change their degrees.
The change of the ABC index is  bounded from above by
\begin{eqnarray} \label{eq-pro-P3-120}
&&-f(d(y),d(y_p))+f(d(y)-1, d(y_p))  \nonumber \\
&& +\sum_{i=1}^{d(y)-3}(-f(d(y),d(x_i))+f(d(y)-1, d(x_i))  \nonumber \\
&&-f(d(y),d(x))+f(d(y)-1,d(x)+d(w)-1)  \nonumber \\
&&-f(d(y),d(w))+f(5,2)  \nonumber \\
&&-f(d(x),4)+f(d(x)+d(w)-1,5)  \nonumber \\
&&+(d(w)-1)(-f(d(w),4)+f(d(x)+d(w)-1,4))  \nonumber \\
&&+(d(x)-2)(-f(d(x),4)+f(d(x)+d(w)-1,4)). 
\end{eqnarray}
First, we consider the case $d(x) > 4$.
By Proposition~\ref{appendix-pro-030-2}, $-f(d(y),d(y_p))+f(d(y)-1, d(y_p))$  increases in $y_p$,
and it is bounded from above by $\lim_{d(y_p) \to \infty}=-f(d(y),d(y_p))+f(d(y)-1, d(y_p))=-1/\sqrt{d(y)}+1/\sqrt{d(y)-1}$.
By the same proposition, $-f(d(y),d(x_i))+f(d(y)-1, d(x_i))$ increases in $d(x_i)$ and by  Proposition~\ref{pro-P3-10}, 
$d(x_i)$ can be at most $d(x)+1$. Thus, $-f(d(y),d(x_i))+f(d(y)-1, d(x_i))$ has an upper bound for $d(x_i)=d(x)+1$.
By Proposition~\ref{appendix-pro-050}, $(d(w)-1)(-f(d(w),4)+f(d(x)+d(w)-1,4))$ decreases in $d(x)$, so it is maximal for
$d(w)=d(x)$.
The function $f(.,.)$ decrease in both variables. Together with the fact that  $d(w)=d(x)$ or $d(w)=d(x)+1$, we
obtain the following bound on  (\ref{eq-pro-P3-120}):
\begin{eqnarray} \label{eq-pro-P3-120-50}
g_2(d(x), d(y)) = 
&&-\frac{1}{\sqrt{d(y)}}+\frac{1}{\sqrt{d(y)-1}}  \nonumber \\
&& +(d(y)-3)(-f(d(y),d(x)+1)+f(d(y)-1, d(x)+1))  \nonumber \\
&&-f(d(y),d(x))+f(d(y)-1,2d(x)-1)  \nonumber \\
&&-f(d(y),d(x)+1)+f(5,2)  \nonumber \\
&&-f(d(x),4)+f(2d(x)-1,5)  \nonumber \\
&&+(2d(x)-3)(-f(d(x),4)+f(2d(x)-1, 4)).  \nonumber
\end{eqnarray}
By Proposition~\ref{appendix-pro-030}, it follows that 
$-f(d(y),d(x)+1)+f(d(y)-1, d(x)+1)$, $-f(d(y),d(x))+f(d(y)-1,2d(x)-1)$,
$-f(d(x),4)+f(2d(x)-1,5)$ and $-f(d(x),4)+f(2d(x)-1,4)$
increase in $d(x)$.
The $-f(d(y),d(x)+1)$ also increases in $d(x)$,
therfore $g_2(d(x), d(y))$ increases in $d(x)$, too.
A verification shows that the largest $d(x)$, for which $g_2(d(x), d(y))$
is negative, is $d(x)=45$.

For the case $d(x)=4$, recall that  $4 \leq d(w) \leq 6$ holds.
Similarly, as above, also for this particular case,
we can show that the change of the ABC index after applying $\mathcal{T}_2$  is negative.

For $d(x) \in [46, 53]$ we proceed as follows.
We set $d(y_p) \to \infty$, while  then $g_2(d(x), d(y))$ is maximal and
we consider the lower envelope of $g_{1e} (d(x), d(y))$ (resp. $g_{1o}(d(x), d(y)$) and $g_2(d(x), d(y))$.
For $d(x)=46$, $g_{1e}(46, d(y))$ is negative for $d(y) \in [46, 124]$,
while $g_2(46, d(y))$ is negative for  $d(y) \in [46, 55] \cup [69, \infty)$. Thus, for $d(x)=46$,
we can obtain always a negative change of the ABC index.
Similarly, one can make the same conclusion for $d(x)=47, \dots, 53$.

\bigskip

\noindent
{\bf Case $2.$}   {\it The root vertex $z$ of $G$ is at level $5$.}

\noindent
The tree $G$ in Figure~\ref{fig-P3-35-0} illustrates this case.
\begin{figure}[!h]
\begin{center}
\includegraphics[scale=0.75]{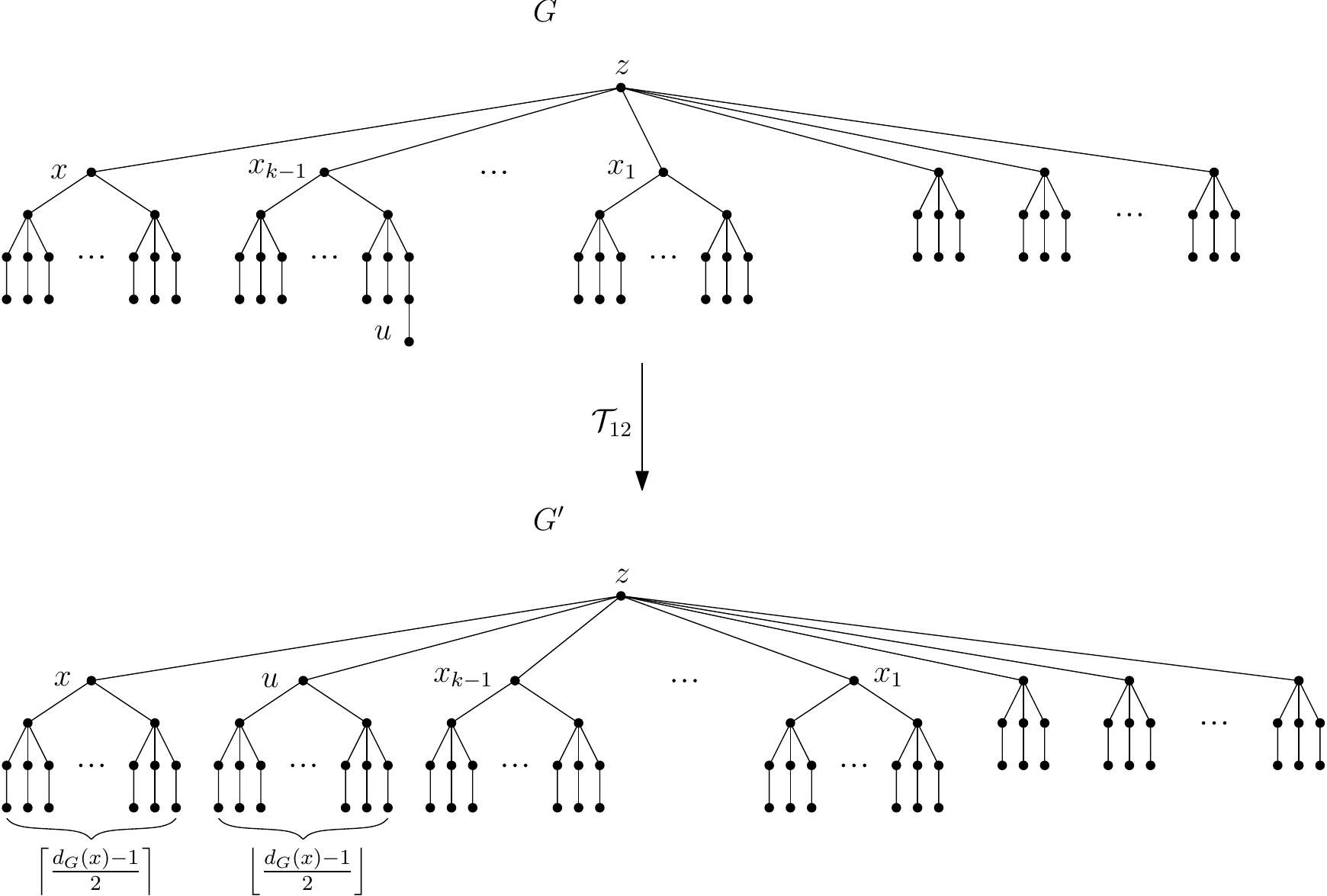}
\caption{An illustration of the transformation $\mathcal{T}_{12}$ Case $2$.}
\label{fig-P3-35-0}
\end{center}
\end{figure}
The number of the children of the root $z$ that are big vertices (vertices $x_1, x_2, \dots, x_{k-1}, x$)
is denoted by $k$.
We assume that $d(x_1) \leq d(x_2) \leq \dots   \leq d(x_{k-1}) \leq d(x)$.

\noindent
First, we apply the transformation   $\mathcal{T}_{12}$ from Case $1.$
In this context, we denote it as $\mathcal{T}_{12}$. An illustration of $\mathcal{T}_{12}$ is given in Figure~\ref{fig-P3-35-0}.
In this case, the change of the ABC index is
\begin{eqnarray} \label{eq-pro-P3-160}
g_{12}(d(x), d(z), d(x_i),k )&=&-f(d(z),d(x))+f \left(d(z)+1,\left\lceil \frac{d(x)-1}{2} \right \rceil  +1\right)   \nonumber \\
&&-f(2,1)+f\left(d(z)+1,\left\lfloor \frac{d(x)-1}{2} \right \rfloor  +1\right)  \nonumber \\
&& +\sum_{i=1}^{k-1}(-f(d(z),d(x_i))+f(d(z)+1, d(x_i))  \nonumber \\
&& +(d(z)-k)(-f(d(z),4)+f(d(z)+1, 4))  \nonumber  \\
&& + \left( \left\lceil \frac{d(x)-1}{2} \right \rceil \right) \left(-f\left(d(x),4\right) + f\left(\left\lceil \frac{d(x)-1}{2} \right \rceil  +1,4\right) \right) \nonumber \\
&& + \left( \left\lfloor \frac{d(x)-1}{2} \right \rfloor \right)  \left(-f\left(d(x),4\right) + f\left(\left\lfloor \frac{d(x)-1}{2} \right \rfloor  +1,4\right) \right). \nonumber   
\end{eqnarray}
By Proposition~\ref{appendix-pro-030}, the expression $-f(d(z),d(x_i))+f(d(z)+1, d(x_i)$ 
decreases in $d(x_i)$. 
Since $d(x_i) \geq 4$, it follows that  $-f(d(z),d(x_i))+f(d(z)+1, d(x_i) \leq -f(d(z),4)+f(d(z)+1, 4)<0$,
and consequently $g_{12}(d(x), d(z), d(x_i),k)$ is maximal when $k$ is minimal, i.e., $k=1$  and
$g_{12}(d(x), d(z), d(x_i),k)$ is bounded from above by
\begin{eqnarray} \label{eq-pro-P3-170}
&&-f(d(z),d(x))+f \left(d(z)+1,\left\lceil \frac{d(x)-1}{2} \right \rceil  +1\right)   \nonumber \\
&&-f(2,1)+f\left(d(z)+1,\left\lfloor \frac{d(x)-1}{2} \right \rfloor  +1\right)  \nonumber \\
&& +(d(z)-1)(-f(d(z),4)+f(d(z)+1, 4)) \nonumber  \\
&& + \left( \left\lceil \frac{d(x)-1}{2} \right \rceil \right) \left(-f\left(d(x),4\right) + f\left(\left\lceil \frac{d(x)-1}{2} \right \rceil  +1,4\right) \right) \nonumber \\
&& + \left( \left\lfloor \frac{d(x)-1}{2} \right \rfloor \right)  \left(-f\left(d(x),4\right) + f\left(\left\lfloor \frac{d(x)-1}{2} \right \rfloor  +1,4\right) \right). 
\end{eqnarray}
Analogous analysis as in Case $1.$ shows that (\ref{eq-pro-P3-170}) is negative for    $d(x) \geq 66$.
Similarly as for $k=1$, it can be derived that for $k=2$ the function $g_{12}(d(x), d(z), d(x_i),k)$ is
negative for $d(x) \geq 65$. Because $g_{12}(d(x), d(z), d(x_i),k)$ is decreasing in $k$, it follows that
it is negative for $d(x) \geq 65$ and $k \geq 2$.
In Table~\ref{T1-parameters} some pairs of values of parameters $d(x)$ and $k$ are given for which  $g_{12}(d(x), d(z), d(x_i),k)$
is negative.

\begin{table}
\caption{Pairs of values of the parameters $d(x)$ and $k$ for which   the change of the ABC index  (the expression (\ref{eq-pro-P3-160})),
after applying  $\mathcal{T}_{12}$ in Case~$2.$, is negative.}
\begin{center}
\begin{tabular}{cccccc}
 \toprule
\multicolumn{6}{c}{$(d(x), k)$} \\
 \midrule
($\geq 66, \, \geq 1$)  & ($\geq 65, \, \geq 2$)   & ($\geq 64, \, \geq 5$)   & $\geq 63, \, \geq 7$)    & ($\geq 62, \, \geq 9$)  & ($\geq 61, \, \geq 11$)  \\

($\geq 60, \, \geq 13$)  & ($\geq 59, \, \geq 15$)    & ($\geq 58, \, \geq 17$)  & ($\geq 57, \, \geq 19$)  & ($\geq 56, \, \geq 23$) & ($\geq 55, \, \geq 29$)   \\

($\geq 54, \, \geq 52$)  &  \multicolumn{5}{c}{}  \\
 \bottomrule
\end{tabular}
\end{center}
\label{T1-parameters}
\end{table}

\noindent
Further, we show, that  a negative change of the ABC index for $d(x) \leq 65$ can be obtained. We distinguish two cases with respect to $k$: $k \geq 2$ and $k = 1$.
Note that if $k = 0$, then $z$ is at level $4$, which is analyzed in Case $3.$

\smallskip

\noindent
{\bf Subcase $2.1.$}  {\it $k \geq 2$.}

\noindent
Here, we apply the transformation   $\mathcal{T}_2$ from Case $1.$
For this subcase we denote it as $\mathcal{T}_{22}$ and illustrate it in Figure~\ref{fig-P3-30-0}. 
\begin{figure}[H]
\begin{center}
\includegraphics[scale=0.75]{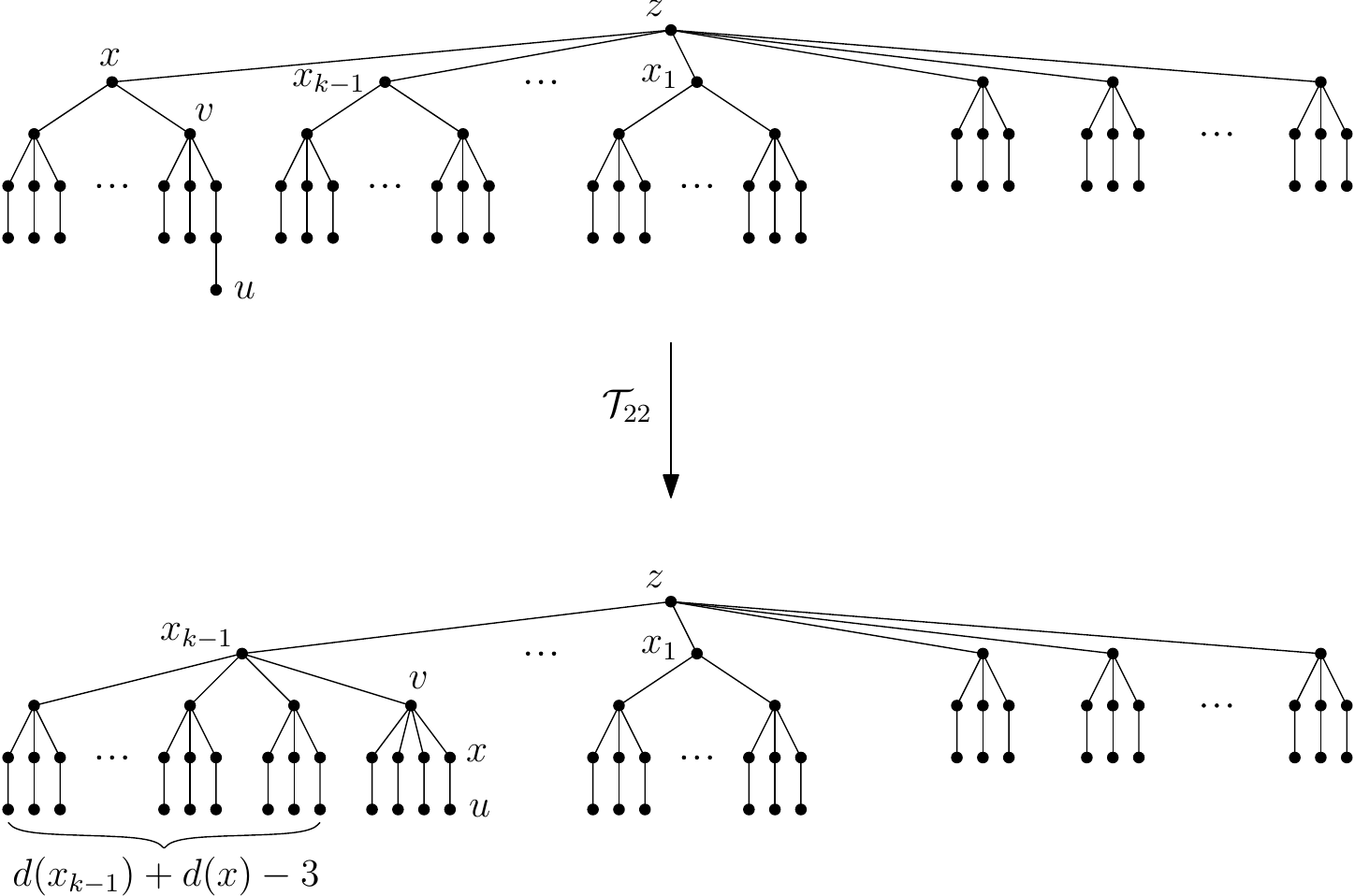}
\caption{The transformation $\mathcal{T}_{22}$ from the proof of  Theorem~\ref{te-P3-20}, Subcase $2.1.$}
\label{fig-P3-30-0}
\end{center}
\end{figure}
The change of the ABC index is
\begin{eqnarray} \label{eq-pro-P3-165}
g_{22}(d(x), d(z), d(x_i),k) = 
&& \sum_{i=1}^{k-2}(-f(d(z),d(x_i))+f(d(z)-1, d(x_i))  \nonumber \\
&& +(d(z)-k)(-f(d(z),4)+f(d(z)-1, 4)  \nonumber \\
&&-f(d(z),d(x_{k-1}))+f(d(z)-1,d(x_{k-1})+d(x)-1)  \nonumber \\
&&-f(d(z),d(x))+f(5,2)  \nonumber \\
&&-f(d(x),4)+f(d(x_{k-1})+d(x)-1,5)  \nonumber \\
&&+(d(x_{k-1})-1)(-f(d(x_{k-1}),4)+f(d(x_{k-1})+d(x)-1,4))  \nonumber \\
&&+(d(x)-2)(-f(d(x),4)+f(d(x_{k-1})+d(x)-1,4)).  \nonumber 
\end{eqnarray}
By Proposition~\ref{pro-P3-20}, we have that the possible values for $d(x_{k-1})$ are $d(x)$, $d(x)-1$ and $d(x)-2$.
For each of these three possible values for $d(x)$, with almost identically analysis as for $g_2(d(x), d(y), d(y_p),d(x),k)$, 
we can show that $g_{22}(d(x), d(z), d(x_i),k)$ is largest for $d(x_{k-1})= d(x)$, it increases in $k$ and is negative for $d(x) \leq 47$. 
Also, we can obtain that for the values of parameters $d(x)$ and $k$ given in Table~\ref{T2-parameters},
$g_{22}(d(x), d(z), d(x_i),k)$ is negative.
\begin{table}
\caption{
Pairs of values of the parameters $d(x)$ and $k$ for which   the change of the ABC index,
after applying  $\mathcal{T}_{22}$ in Subcase~$2.2.$ is negative.}
\begin{center}
\begin{tabular}{ccc}
\toprule
\multicolumn{3}{c}{$(d(x), k)$} \\
\midrule
($\leq 48$,   $\leq d(z) - 3$)  & 
($\leq 49$,   $\leq d(z) - 8$)  &
($\leq 50$,   $\leq d(z) - 17$)  \\
($\leq 51$,   $\leq d(z) - 36$)  &
($\leq 52$,   $\leq d(z) - 87$)  &
($\leq 53$,   $\leq d(z) - 516$)  \\
\bottomrule
\end{tabular}
\end{center}
\label{T2-parameters}
\end{table}

Now, we consider the pairs of $d(x)$ and $k$  for which we did obtain that $g_{22}(d(x), d(z), d(x_i),k)$
is negative: $d(x)=48$,   $k > d(z) - 3$;  $d(x)=49$,   $k > d(z) - 8$; $d(x)=50$,   $k > d(z) - 17$; $d(x)=51$,   $k > d(z) - 36$; $d(x)=52$,   $k > d(z) - 87$; $d(x)=53$,   $k > d(z) - 516$. 
In this case we take the lower envelope of $g_{22}(d(x), d(z), d(x_i),k)$ and $g_{12}(d(x), d(z),$ $d(x_i),k )$. 
Due to Proposition~\ref{pro-P3-10}, in these cases the two different possibilities
of $d(x_i)$ are $d(x_i)=d(x)$ and $d(x_i)=d(x)-1$.

For $d(x)=48$ and $k >d(z)-3$, we consider first $k=d(z)-2$. The possible values for $d(x_i)$  are $48$ and $47$.
Functions $g_{22}(48, d(z), 48,d(z)-2)$ and $g_{12}(48, d(z), 48,d(z)-2 )$, 
$g_{22}(48, d(z), 47,d(z)-2)$ and $g_{12}(48,$ $ d(z), 47,d(z)-2 )$  are depicted in Figure~\ref{fig-Table2-10} $(a)$. Recall that $d(z) \geq d(x)$
Their lower envelope of of both function has always negative values. 
For  $k =d(z)-1$ and  $k =d(z)$, we proceed identically  as for $k =d(z)-2$
and obtain that the corresponding lower envelopes of the functions $g_{22}(d(x), d(z), d(x_i),k)$ and $g_{12}(d(x,) d(z), d(x_i),k )$ 
are negative (Figure~\ref{fig-Table2-10} $(b)$ and $(c)$).

Similarly, we obtain that lower envelops in the cases $d(x)=49$ and $k> d(z) - 8$, $d(x)=50$ and $k> d(z) - 17$,
$d(x)=51$ and $k> d(z) - 36$,  $d(x)=52$ and  $k > d(z) - 87$, and $d(x)=53$ and   $k > d(z) - 516$ are negative.  
Observe that in the last two cases,  $d(z) \geq 88$ (resp., $d(z) \geq 517$), since $k >1$.
Since we apply the same argument, we omit the repetition of the analysis in these cases.

\begin{figure}[!http]
\begin{center}
\includegraphics[scale=0.8]{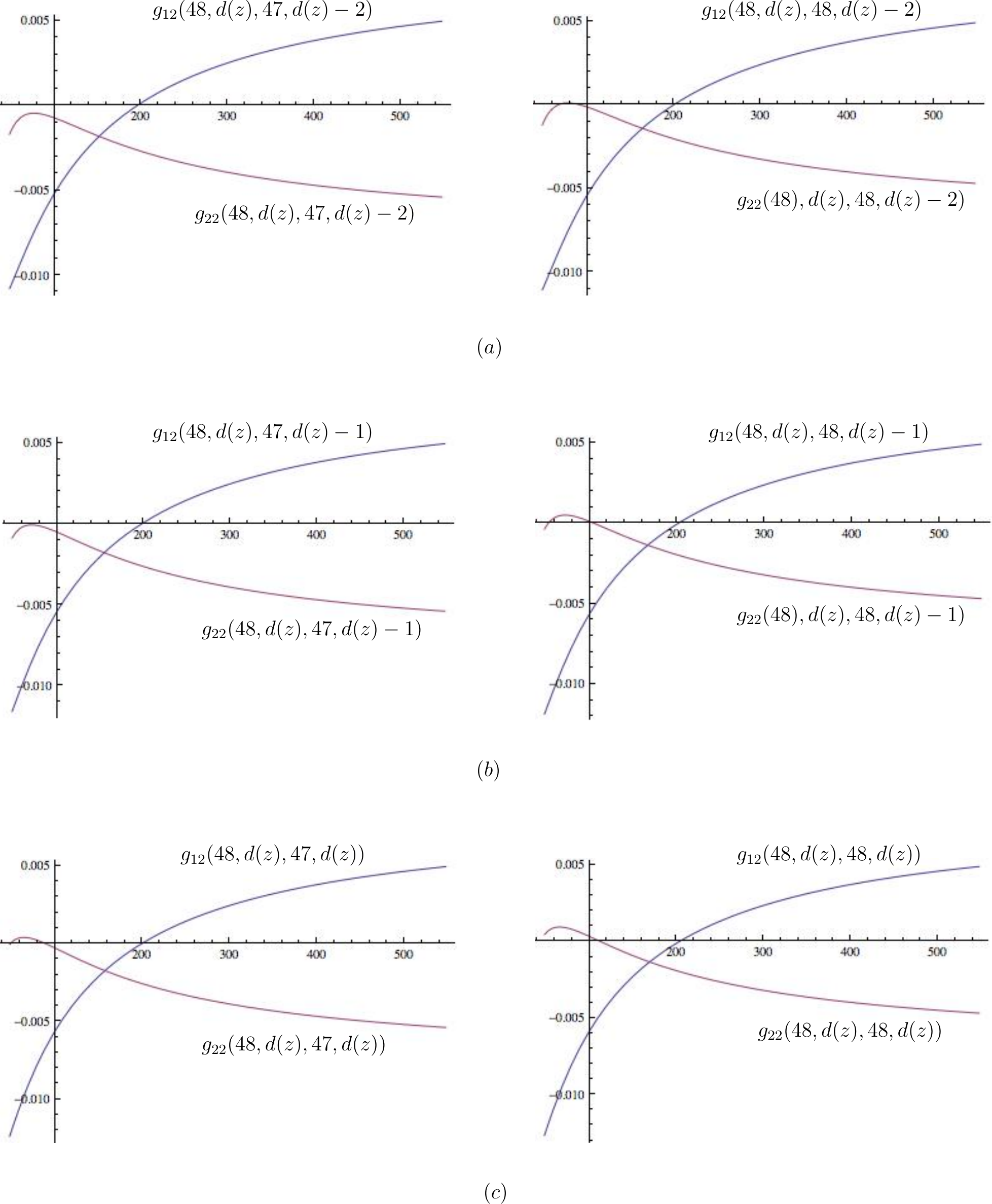}
\caption{Functions  $g_{12}(d(x), d(z), d(x_i),k)$ and $g_{22}(d(x) d(z), d(x_i),k )$ for $d(x)=48$, $d(x_i)=d(x)-1, d(x)$, 
and $k=d(z)-2, d(z)-1, d(z)$. 
The lower envelope of an appropriate pair of functions is negative, i.e., we can obtain negative change of the ABC index of $G$ for the above values of 
$d(x)$, $d(x_i)$ and $k$ by combining transformations $\mathcal{T}_{12}$ and $\mathcal{T}_{22}$.}
\label{fig-Table2-10}
\end{center}
\end{figure}

It remains to prove the theorem for $54 \leq d(x) \leq 65$ and the values of the parameter 
$k$ for which  $g_{22}(d(x), d(z), d(x_i),k)$ is not negative  (see table Table~\ref{T1-parameters}). 
Again, 
by Proposition~\ref{pro-P3-20}, we have that the possible values for $d(x_{k-1})$ are $d(x)$, $d(x)-1$ and $d(x)-2$.
We consider the lower envelope of $g_{22}(d(x), d(z), d(x_i),k)$ and $g_{12}(d(x), d(z), d(x_i),k )$. 
For example, for $d(x)=54$ and $k <52$, we consider the functions 
$g_{22}(54, d(z), 53, k)$ and $g_{12}(54, d(z), 53, k)$, and 
$g_{22}(54, d(z), 54, k)$ and $g_{12}(54, d(z), 54, k)$, for  $k = 2, \dots, 51$, whose lower envelope is negative.
Similarly, we proceed with $d(x)= 55, \dots, 65$. Due to similarity with previous cases, we omit the details here.

Next, we show that the theorem is true for $k = 1$ and $d(x) \leq 65$.

\smallskip

\noindent
{\bf Subcase $2.2.$}  {\it $k = 1$.} 
\begin{figure}[!h]
\begin{center}
\includegraphics[scale=0.75]{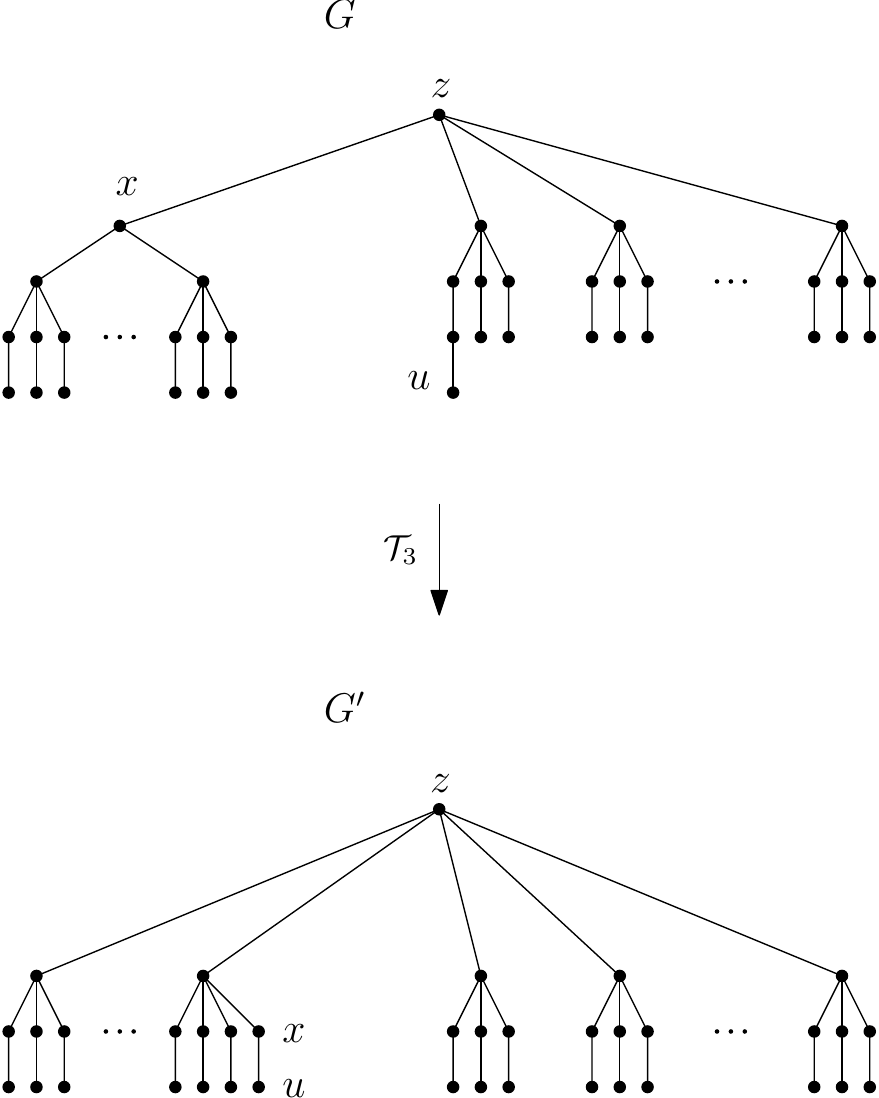}
\caption{The transformation $\mathcal{T}_{3}$ from the proof of  Theorem~\ref{te-P3-20}, Subcase $2.2$.}
\label{fig-P3-35}
\end{center}
\end{figure}

\noindent
First, we apply the transformation $\mathcal{T}_3$ depicted in Figure~\ref{fig-P3-35}.
After applying $\mathcal{T}_3$ the degree of the vertex $z$ increases by $d(x)-2$,
the degree of one child of $x$  increases by $1$, and 
the degree of $x$ decreases to $2$.
The rest of the vertices do not change their degrees.
The change of the ABC index is 
\begin{eqnarray} \label{eq-pro-P3-140}
g_3(d(x), d(z)) = 
&& (d(x)-2)(-f(d(x),4)+f(d(z)+d(x)-2, 4))  \nonumber \\
&&-f(d(x),4)+f(d(z)+d(x)-2, 5) \nonumber \\
&&-f(d(z),d(x))+f(5,2)  \nonumber \\
&& +(d(z)-1)(-f(d(z),4)+f(d(z)+d(x)-2, 4))  \nonumber 
\end{eqnarray}

\noindent
For any $d(x) \in [4, 31]$ the function $g_{3}(d(x), d(z))$ does not have real roots and is negative.
It follows that for $4 \leq d(x) \leq 31$, there exists a tree with smaller ABC index than the assumed minimal one, which is a contradiction
to the initial assumption that $G$ is a tree with minimal ABC index.
It can be also verified that for any $d(x), d(z)$  that satisfy  $4 \leq d(x) \leq d(z) \leq 62$, the function $g_{3}(d(x), d(z))$ is negative.

For $63 \leq d(x) \leq 65$, we apply the transformation $\mathcal{T}_{4}$ depicted in
Figure~\ref{fig-P3-38}.
\begin{figure}[!h]
\begin{center}
\includegraphics[scale=0.75]{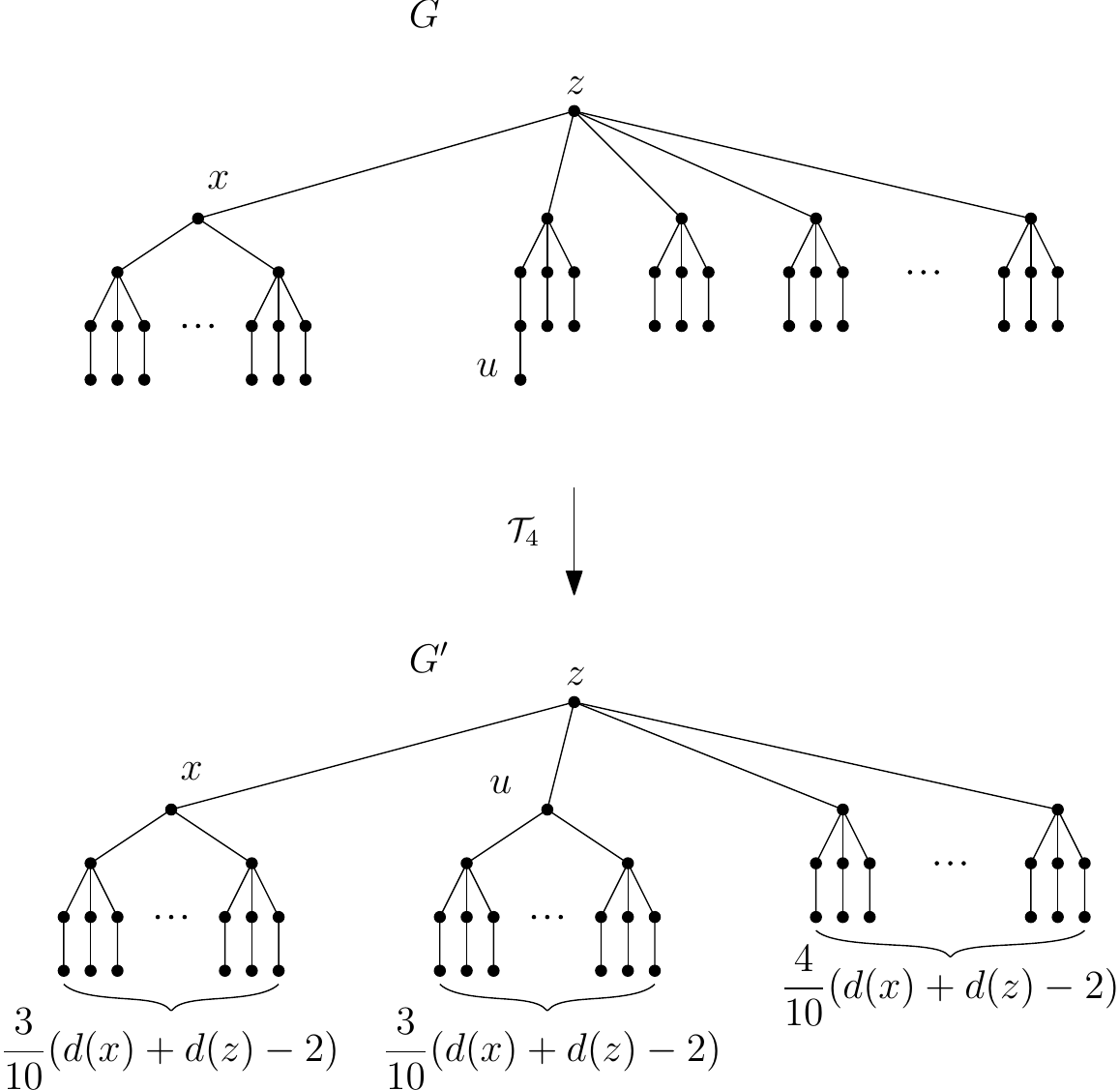}
\caption{The transformation $\mathcal{T}_{4}$ from the proof of  Theorem~\ref{te-P3-20}, Subcase $2.2$.}
\label{fig-P3-38}
\end{center}
\end{figure}
The change of the degrees of vertices in $G$, after applying $\mathcal{T}_{4}$, are as follows:
$d_{G'}(x)=3(d(x)+d(z)-2)/10+1$,  $d_{G'}(u)=3(d(x)+d(z)-2)/10+1$ and 
$d_{G'}(z)=4(d(x)+d(z)-2)/10+2$.
With respect to degree $d(x)$, we distinguish two cases: $d_{G'}(x) <d(x)$ 
and $d_{G'}(x) \geq d(x)$.

When $d_{G'}(x) < d(x)$, the change of the ABC index is
\begin{eqnarray} \label{eq-pro-P3-38-10}
g_{41}(d(x), d(z)) &=&-f(d(x),d(z))+f\left(d_{G'}(x),d_{G'}(z)\right)  \nonumber \\
&&-f(2,1)+f\left(d_{G'}(u),d_{G'}(z)\right)  \nonumber \\
&&+(d_{G'}(x)-1)\left(-f(d(x),4)+f\left(d_{G'}(x),4\right)\right)  \nonumber \\
&&+\left(d(x)-d_{G'}(x)\right)\left(-f(d(x),4)+f\left(d_{G'}(u),4\right)\right)  \nonumber \\
&&+(d_{G'}(z)-2) \left(-f(d(z),4)+f\left(d_{G'}(z),4\right)\right)  \nonumber \\
&&+\left(d(z) -d_{G'}(z)+1\right)\left(-f(d(z),4)+f\left(d_{G'}(u),4\right)\right).  \nonumber 
\end{eqnarray}
Recall that when the index of the degree is omitted, it is always assumed that it is $G$.
From $3(d(x)+d(z)-2)/10 +1 < d(x)$, it follows that $d(z)< (7d(x)-4)/3$. Also, recall that 
$d(z) \geq d(x)$.
It can be verified that for every $d(z) \in [d(x),(7d(x)-4)/3)$ and for $d(x) \in \{63, 64, 65\}$,
the function $g_{41}(d(x), d(z))$ is negative and decreasing in $d(z)$.

When $d_{G'}(x)  \geq d(x)$, from which follows $d(z) > (7d(x)-4)/3)$, the change of the ABC index is
\begin{eqnarray} \label{eq-pro-P3-38-20}
g_{42}(d(x), d(z)) &=&-f(d(x),d(z))+f\left(d_{G'}(x),d_{G'}(z)\right)  \nonumber \\
&&-f(2,1)+f\left(d_{G'}(u),d_{G'}(z)\right)  \nonumber \\
&&+(d(x)-1)\left(-f(d(x),4)+f\left(d_{G'}(x),4\right)\right)  \nonumber \\
&&+\left(d_{G'}(x)-d(x)\right)\left(-f(d(z),4)+f\left(d_{G'}(x),4\right)\right)  \nonumber \\
&&+(d_{G'}(z)-2) \left(-f(d(z),4)+f\left(d_{G'}(z),4\right)\right)  \nonumber \\
&&+\left(d_{G'}(u)-1\right)\left(-f(d(z),4)+f\left(d_{G'}(u),4\right)\right). \nonumber 
\end{eqnarray}
For $d(z) > (7d(x)-4)/3)$ and for $d(x) \in \{63, 64, 65\}$ 
the function $g_{42}(d(x), d(z))$ does not have real roots and it is negative.

It remains to show that for $32 \leq d(x) \leq 62$ the theorem holds.
To show this, we apply the transformations  $\mathcal{T}_5$, 
under the constrain  that  $d(z) > 62$ (recall that for $d(z) \leq 62$, the function $g_{3}(d(x), d(z))$ is negative for any $d(x) > 4$).
The transformation  $\mathcal{T}_{5}$ depicted in Figure~\ref{fig-P3-36}.
\begin{figure}[!h]
\begin{center}
\includegraphics[scale=0.75]{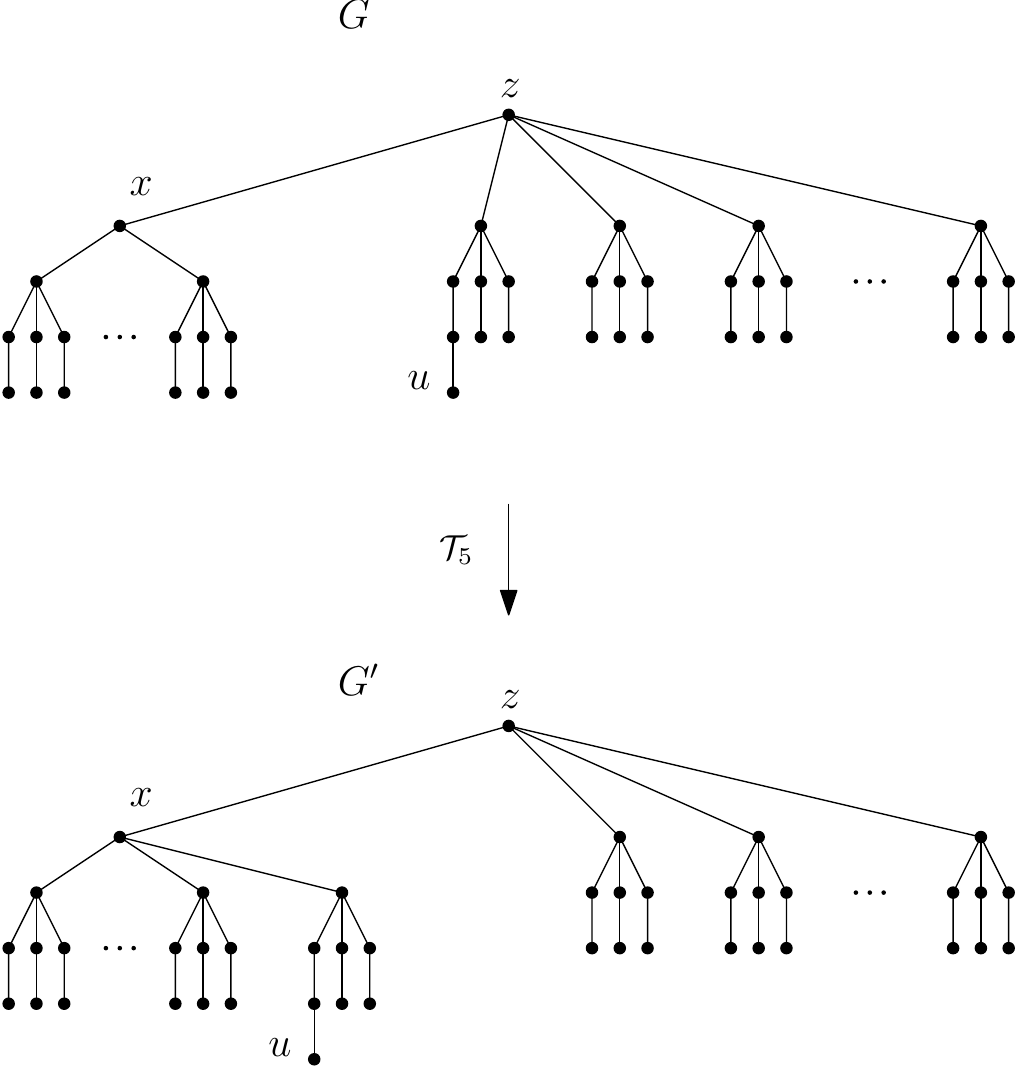}
\caption{The transformation $\mathcal{T}_{5}$ from the proof of  Theorem~\ref{te-P3-20}, Subcase $2.2$.}
\label{fig-P3-36}
\end{center}
\end{figure}
After this transformation the only vertices that change there degrees are $z$ and $x$,
namely, $z$ decreases its degree for $1$ and $x$ increases its degree for $1$.
The change of the ABC index after applying $\mathcal{T}_{5}$ is
\begin{eqnarray} \label{eq-pro-P3-165}
g_{5}(d(x), d(z)) = 
&&-f(d(z),d(x))+f(d(z)-1,d(x)+1)  \nonumber \\
&&-f(d(z),4)+f(d(x)+1,4)  \nonumber \\
&& +(d(x)-1)(-f(d(x),4)+f(d(x)+1, 4)  \nonumber \\
&& +(d(z)-2)(-f(d(z),4)+f(d(z)-1, 4) ) \nonumber 
\end{eqnarray}
Consider first the case $d(x)=32$. The function $g_{5}(32, d(z))$, for $d(z) > 62$,
does not have real roots and is negative.
Similarly, we verify that $g_{5}(d(x), d(z))$ is negative for $d(x)=33, \dots, 61$ and $d(z) > 62$.
For $d(x)=62$, $g_{5}(d(x), d(z))$ is negative for $d(z) > 63$ and for $d(z) =63$ it is zero.
Here, for $d(x)=62$ and $d(z)=63$ we apply the transformation $\mathcal{T}_{3}$.
It holds that $g_3(62, 63)=-0.0000277276$, and hence, we have shown that also for $32 \leq d(x) \leq 62$ we can
obtain a negative change of the ABC index, which concludes the proof of Case~$2$.

\bigskip

\noindent
{\bf Case $3.$}   {\it The root vertex $z$ of $G$ is at level $4$.} 
\bigskip

\noindent
In this case $G$ is an improper Kragujevac tree.
Here we apply the transformation $\mathcal{T}_6$ illustrated in Figure~\ref{fig-P3-40}.

\begin{figure}[!h]
\begin{center}
\includegraphics[scale=0.75]{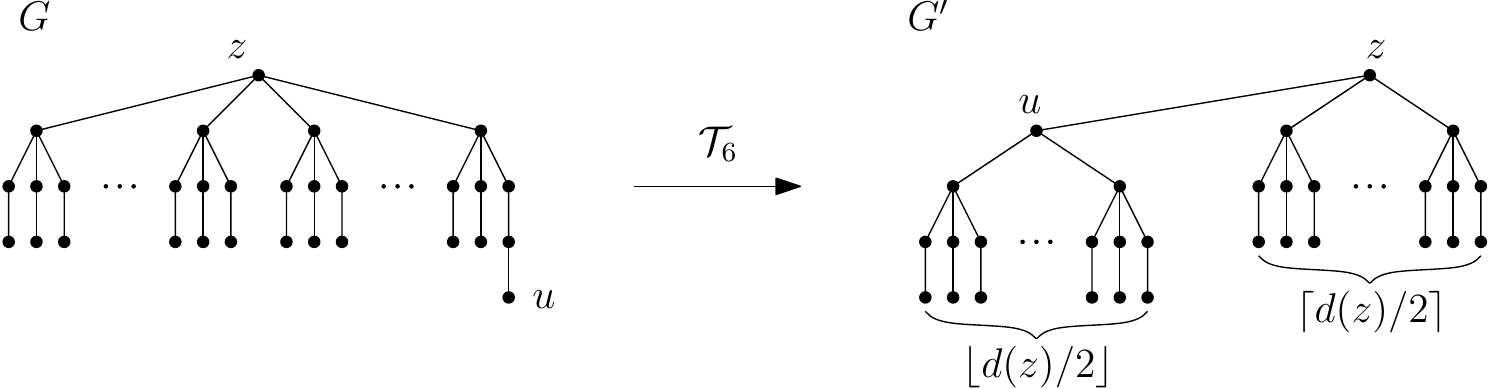}
\caption{The transformation $\mathcal{T}_6$ from the proof of  Theorem~\ref{te-P3-20}, Case $3$.}
\label{fig-P3-40}
\end{center}
\end{figure}
After applying $\mathcal{T}_6$ the degree of the vertex $z$ decreases to $\lfloor d(z)/2 \rfloor$,
while the degree of $u$ increases to $\lceil d(z)/2 \rceil$.
The rest of the vertices do not change their degrees.
The change of the ABC index is
\begin{eqnarray*} \label{eq-Case-3}
g_6(d(z)) &=&  \lceil d(z)/2 \rceil (-f(d(z),4)+f(\lceil d(z)/2 \rceil +1, 4))  \nonumber \\
&& + \lfloor d(z)/2 \rfloor (-f(d(z),4)+f(\lfloor d(z)/2 \rfloor  +1, 4))  \nonumber \\
&& -f(2,1)+f(\lceil d(z)/2 \rceil +1,\lfloor d(z)/2 \rfloor  +1))  \nonumber 
\end{eqnarray*}

For even $d(z)$, the function $g_6(d(z))$ has only one zero at $d(z) \approx 59.5903$, and for $d(z) > 59.5903$ is negative.
For odd $d(z)$, $g_6(d(z))$ has only one zero at $d(z)  \approx 59.6067$, and it is negative for $d(z) > 59.6067$.
Since $d(z)$ is an integer, it follows that in this case the change of the $ABC$ index for $d(z) > 59$ (or, $n > 415$) is negative.
This concludes the proof of the theorem. 

\end{proof}

\section[Appendix]{Appendix}\label{sec:Appendix}

The next  propositions are used in the proofs of the previous section.
The function $f(x,y)$ is defined as in (\ref{eqn:001}).

\begin{pro}[\cite{d-sptmabci-2014}] \label{appendix-pro-030}
Let $g(x,y)=-f(x,y)+f(x+\Delta x,y-\Delta y)$, with  real numbers $x, y \geq 2$,  $\Delta x \geq 0$,  $0 \leq \Delta y < y$.
Then, $g(x,y)$ increases in $x$ and decreases in $y$.
\end{pro}

\noindent
Due to the symmetry of the function $f(.,.)$,  Proposition~\ref{appendix-pro-030} can be rewritten as follows.

\begin{pro} [\cite{d-sptmabci-2014}] \label{appendix-pro-030-2}
Let $g(x,y)=-f(x,y)+f(x -\Delta x,y+\Delta y)$, with  real numbers $x, y \geq 2$,  $\Delta y \geq 0$,  $0 \leq \Delta x < x$.
Then, $g(x,y)$ decreases in $x$ and increases in $y$.
\end{pro}

\begin{pro}\label{appendix-pro-040}
Let  $g(x)=(x-2)(-f(x,4)+f(x-1,4))$ with  real number $x \geq 4$.
Then, $g(x)$ decreases in $x$.
\end{pro}
\begin{proof} 
The first derivative of $g(x)$ after simplification and rearranging is
\begin{eqnarray*} \label{eq-pro-P3-50}
\frac{\partial g(x)}{\partial x} 
&=&\frac{\left(-2+3 x-x^2+x^3-x^4\right) \ds \sqrt{\frac{1+x}{-1+x}}+\left(x^2-x^3+x^4 \right)  \ds \sqrt{\frac{2+x}{x}}}{2 (-1+x)^2 x^2  \ds \sqrt{\frac{1+x}{-1+x}}  \ds  \sqrt{\frac{2+x}{x}}}. \nonumber
\end{eqnarray*}
To show that $\partial g(x)/ \partial x$ is negative, suffices to show
\begin{eqnarray} \label{eq-pro-P3-60}
&&\left(-2+3 x-x^2+x^3-x^4\right) \sqrt{\frac{1+x}{-1+x}}< -\left(x^2-x^3+x^4 \right)  \ds \sqrt{\frac{2+x}{x}}. \nonumber 
\end{eqnarray}
Quadrating and simplifying the last expression gives
\begin{eqnarray} \label{eq-pro-P3-70}
4 - 4 x - 3 x^2 + 2 x^3 - 2 x^4 + 8 x^5 - 2 x^6 <0 \nonumber
\end{eqnarray}
which is satisfied for $x \geq 4$.
\end{proof} 

\noindent
The following result, Proposition~\ref{appendix-pro-050}, is similar to that in Proposition~\ref{appendix-pro-040}.
Since, the proofs are very similar we omit the proof of Proposition~\ref{appendix-pro-050}.
\begin{pro}\label{appendix-pro-050}
Let  $g(x)=(x-1)(-f(x,4)+f(x+c-1,4))$ with  real number $x \geq 4$ und $c$ is a positive real constant.
Then, $g(x)$ decreases in $x$.
\end{pro}

%
%

%

\end{document}